\documentclass[10pt, conference, letterpaper]{IEEEtran}

 \usepackage{umoline}
 \usepackage{array}
 \usepackage{amssymb,epsfig,latexsym,epsf,color}
 \usepackage{algorithm, algorithmicx, algpseudocode}
 \usepackage{graphicx,amsmath,amssymb,amsthm,latexsym}
 \usepackage[hang]{subfigure}
 \usepackage{paralist}
 \usepackage{bm}
 \usepackage{times}
 \usepackage{multirow}
 \usepackage{pseudocode}
 \usepackage{balance}
 \usepackage{color}
 \usepackage{cite}
 \usepackage{comment}
 \usepackage[hyphens]{url}
 \usepackage{mathtools,nccmath}
 \usepackage{bbm}
 \usepackage{booktabs}
 \usepackage{pbalance}

\IEEEoverridecommandlockouts
\allowdisplaybreaks

\renewcommand{\baselinestretch}{0.98}

\DeclareGraphicsExtensions{.pdf,.png,.jpg,.eps}
\usepackage{overpic} 

\newcommand{\set}[1]{\ensuremath{\mathcal #1}}

\newcommand{\separator}{
  \begin{center}
    \rule{\columnwidth}{0.3mm}
  \end{center}
}

\newtheorem{proposition}{Proposition}[section]
\newtheorem{lemma}{Lemma}[section]
\newtheorem{definition}{Definition}[section]

\newcommand{\beq}{\begin{eqnarray*}}
\newcommand{\eeq}{\end{eqnarray*}}
\newcommand{\beqn}{\begin{eqnarray}}
\newcommand{\eeqn}{\end{eqnarray}}
\newcommand{\bemn}{\begin{multiline}}
\newcommand{\eemn}{\end{multiline}}

\def\A{\mathcal A}

\newtheorem{theorem}{Theorem}[section]

\DeclareFixedFont{\BX}{U}{cmss}{bx}{n}{10}

\usepackage{etoolbox}
\makeatletter
\patchcmd{\maketitle}{\@copyrightspace}{}{}{}
\makeatother

\def\BibTeX{{\rm B\kern-.05em{\sc i\kern-.025em b}\kern-.08em
    T\kern-.1667em\lower.7ex\hbox{E}\kern-.125emX}}

\begin{document}

\title{Economics of Spot Instance Service:\\A Two-stage Dynamic Game Approach}

\author{\IEEEauthorblockN{Hyojung Lee$^\dag$, Lam Vu$^*$, and Minsung Jang$^\dag$}
  \IEEEauthorblockA{\textit{$^\dag$Cloud Research Team, Samsung SDS}\\
    }
    \IEEEauthorblockA{\textit{$^*$Samsung SDSV}\\
    }
    }


\maketitle

\begin{abstract}
 This paper presents the economic impacts of spot instance service on the cloud service providers (CSPs) and the customers when the CSPs offer it along with the on-demand instance service to the customers. We model the interaction between CSPs and customers as a non-cooperative two-stage dynamic game. Our equilibrium analysis reveals (i) the techno-economic interrelationship between the customers' heterogeneity, resource availability, and CSPs' pricing policy, and (ii) the impacts of the customers' service selection (spot vs. on-demand) and the CSPs' pricing decision on the CSPs' market share and revenue, as well as the customers' utility. The key technical challenges lie in, first, how we capture the strategic interactions between CSPs and customers, and second, how we consider the various practical aspects of cloud services, such as heterogeneity of customers' willingness to pay for the quality of service (QoS) and the fluctuating resource availability. The main contribution of this paper is to provide CSPs and customers with a better understanding of the economic impact caused by a certain price policy for the spot service when the equilibrium price, which from our two-stage dynamic game analysis, is able to set as the baseline price for their spot service.
\end{abstract}




\section{Introduction}
Cloud service providers (CSPs) offer their customers with a highly cost-effective service utilizing CSPs' unused resources. The service, which is usually referred to as spot instance service, is beneficial to both CSPs and customers as the CSPs can exploit those unused resources for a new revenue source while the customers can reduce their cloud cost. Although the customer's spot instances can be terminated at any time by CSPs' policies, the spot service has been growing in popularity because it~\cite{AWSSpot, AzureSpot, IBMSpot, GCPSpot} is offered much less price than their regular on-demand service.



It, however, has not been well-studied the economic impacts of spot service from both service providers' and customers' perspective. This paper argues that CSPs and customers must understand their strategic interactions so that provider's pricing policy for individual services, spot and on-demand, is set accordingly in the cloud service market. 
As mentioned, an on-demand instance will be available until its customer terminates it, which comes at a higher price while spot will not offer that level of service availability, but be much cheaper. Under these circumstances, the customers have to choose which service will maximize their utility, and the CSPs should set their pricing policies for each service to increase revenue. For a mutually beneficial cloud service market of all participants, it is important to understand the strategic interaction between CSPs' pricing policies for on-demand vs. spot services and customers' service selection.


Predicting how a certain pricing policy for the spot service impacts CSP's total revenue is challenging because on the one hand, the spot service could deliver a positive impact extending their customer base, but on the other hand, it might negatively impact the revenue due to a cheaper price and possible cannibalization of its pricey counterpart, on-demand service~\cite{arrow1954existence}. 
In addition, the availability of resources for the spot service is hard to predict as those become available only after all the on-demand instances do not fully use them. Therefore, CSPs providing spot such as AWS~\cite{AWSSpot}, Google Cloud~\cite{GCPSpot}, IBM Cloud~\cite{IBMSpot}, and Azure~\cite{AzureSpot} do not set their spot prices as a fixed value, but dynamically determine it according to the availability of the resources and the  customers' demand~\cite{AWSSpot}. Due to the nature of resources being limited and fluctuating, these CSPs allow the customers to set a maximum bid similar to an auction when they request spot instances so that it is possible for CSPs to reclaim their resources (i.e., terminating spot instances currently serving customers) if the spot price set by a CSP goes up above the bid.

To clearly understand the economic impacts of the relationship between a pricing policy of CSP for the instance services (spot and on-demand) and a service selection of customers in the cloud service market, this paper formally models the interaction as a non-cooperative two-stage dynamic game, which is a well-known game-theoretic approach~\cite{von2007theory}. The model captures the strategic behaviors of CSPs and customers and how we consider the various practical aspects of cloud services, such as the heterogeneity of customers' willingness to pay with respect to the quality of service (i.e., the availability of their instances) and the unpredictable amount of unused resources eventually available for spot in a data center. By doing so, we aim to answer the following key question: how much and under what conditions the spot instance service affects the economic benefits of CSPs and customers. The main contributions of this paper are as follows:

\begin{compactitem}[$\circ$]
\item We model a two-stage dynamic game to capture the interaction between CSPs and customers in the cloud service market. The game-theoretic analysis results in a closed-form unique equilibrium from which the practical implications of the spot service are derived, e.g., the pricing policy of CSPs and the service selection of customers not only for maximizing profits but also for the stable cloud service market.
\item We quantify the economic benefits of spot instance service on both CSPs and customers, for which we compare the equilibrium analysis results, analytically and numerically, between the case of serving spot instances and that of not serving them, i.e., only on-demand instances exist. Therefore, we show under what condition of the resource utilization and the QoS of each service, how much more profit the spot service can be brought to CSPs and customers.
\item We provide a simple experimental evaluation verifying the profitability of the spot instance service when enough idle resources exist after serving the on-demand instance service for a certain CSP. Furthermore, it is also shown that we need proper provisioning and eviction algorithm for spot instances to avoid interrupting the provisioning of on-demand instances.
\end{compactitem}

\section{Model}

This paper assumes a market where a cloud service provider (CSP) serving two types of instance services, spot and on-demand, and customers can choose which service to use. 









\subsection{System Model}\label{subsec:model}


\smallskip
\noindent{\bf \em Cloud Service Provider (CSP).}
We consider a single cloud service provider which provides an on-demand instance service and a spot instance service to customers. 
There is a set $\set M = \{1, 2, \dots, M\}$ of $M$ computing server clusters, and we index the clusters by $m$. Note that the prices of the on-demand instance service and the spot instance service can vary across the clusters. Let $c^m$ be the maximum capacity of computing resource in cluster $m \in \set M.$ In the real-world cloud market, a variety of computing resource services are offered by multiple CSPs. This paper, however, primarily focuses on analyzing the relationship between spot instance service and on-demand instance service, which have the most considerable impact on spot instance service. Moreover, our model and analysis results can be linearly extended for cases involving multiple CSPs.



\smallskip
\noindent{\bf \em Customers.} There is a set $\set I = \{1, 2, \dots, I\}$ of $I$ customers spread in the whole logical regions (or simply regions throughout this paper) that the entire clusters are covered. We assume that those regions are disjoint, thus we use $i \in \set I^m$ to index a particular customer under the region covered by cluster $m \in \set M.$ This assumption implies a customer uses its service not from multiple clusters but from a single cluster. Note that the elements of subset $\set I^m \subseteq \set I$ for all $m \in \set M$ can be varied according to the customers' willingness to choose a cluster or a service provisioning/migration policy of CSP. Each customer $i \in \set I$ has its computing resource demand $x_i(t)$ at time slot $t \in \set T$ where $\set T = \{1, 2, \dots, T\}.$ We consider an arbitrary time scale for an average analysis over the unit billing cycle.

\smallskip
\noindent{\bf \em On-demand and spot instance services.}
Each customer can have three available options for choosing a service: (i) on-demand instance service, (ii) spot instance service, and (iii) no service. Note that the customers can choose not to use any service if there is no one satisfying their willingness to pay. Denote $\set A$ be the set of services, $\set A = \{o, s, n\},$ where $o, s,$ and $n$ refer to the first character of the services, respectively. 
Thus, a customer pays no money if it chooses no service, i.e., $p_n = 0$. Otherwise, for customers of cluster subset $\set I^m$ at time slot $t \in \set T$, the CSP sets the service prices $p_o^m (t)$ and $p_s^m (t)$. 
Furthermore, each customer $i$ chooses its service at each time slot $t,$ i.e., $a_i(t) \in \set A.$
Similar to notations of prices, $q_o^m (t)$ and $q_s^m(t)$ denote the quality of services (QoSes) of on-demand and spot instance services. Based on service level agreements (SLAs), we assume that the on-demand instance service always guarantees a higher QoS than the spot instance service, i.e., $q_o^m (t) > q_s^m (t),$ for a cluster $m$ at time slot $t.$

\smallskip
\noindent{\bf \em Resource utilization.}
Let $\set I^m_o$ and $\set I^m_s$ be the sets of customers within a subset $\set I^m$ choosing on-demand instance service and spot instance service, respectively. Then, at time slot $t,$ the aggregate resource utilization by on-demand instance service is defined as $x_o^m(t),$ i.e., $x_o^m(t) = \sum_{i \in \set I^m_o} x_i(t).$  Similarly, $x_s^m(t)$ denotes the resource utilization by spot instance service at time slot $t$, i.e., $x_s^m(t) = \sum_{i \in \set I^m_s} x_i(t).$
We assume that the sum of aggregate resource utilization of on-demand and spot instance services is less than or equal to the maximum capacity of cluster $m,$ i.e., $x_o^m (t) + x_s^m (t) \leq c^m.$
This assumption is reasonable because the purpose of CSP offering the spot instance service is to utilize the idle resource of the cluster running the on-demand service.

\subsection{Game Formulation}\label{subsec:game}

A game-theoretic analysis aims at finding equilibrium as a solution, which is a {\em stable} point from which no player
has an incentive to deviate~\cite{nash1951non}.
In other words, at the equilibrium, both the CSPs and the customers do not need to change their strategies (pricing and service selection) to increase their revenues or utilities.
In this section, we develop a game-theoretic model to analyze the impacts of spot instance service on the computing resource market at the equilibrium where we have
the following key question: how much and under what conditions the spot instance service affects the economic benefits of CSP and customers.

\smallskip
\noindent{\bf \em Background of two-stage dynamic game.}
A two-stage dynamic game is a well-known approach for analyzing a staged strategic interaction between two players, which is useful for modeling a situation where the player (i.e., leader) decides its action at the first stage based on the knowledge of the other player (i.e., follower), and subsequently, the follower also chooses its strategy at the second stage~\cite{von2007theory}.
In this paper, we assume that the CSP knows the characteristics of customers' behavior (e.g., a distribution on customers' willingness to pay on the QoS of computing services), which enables us to develop the two-stage dynamic game-theoretic model. At the first stage, the CSP decides how to determine its price strategy for each service, and at the second stage, customers choose which service to subscribe to. 
Note that the equilibrium analysis of this game is composed of two consecutive steps, referred to as {\em backward induction}~\cite{von2007theory}, which will be explained in Section~\ref{sec:analysis}.

\smallskip
\noindent{\bf \em Utility of customers.}
Our model reflects two key attributes of customers, which are how much they can afford for a certain level of quality of service and how much computing resource they demand. 
We introduce the utility function of customer $i$ under the cluster region $\set I^m$ at a time slot $t \in \set T.$ For a given service price $p_o^m(t)$ or $p_s^m(t)$, the strategy of customer $i$ is a service to use $a_i(t) \in \set A$: 
\begin{eqnarray}\label{eq:utility}
u_i(a_i(t)) &=&
\begin{cases}
\theta_i q_o^m(t) x_i(t) - p_o^m(t) x_i(t),&\text{\small if $a_i(t)= o$,}\\
\theta_i q_s^m(t) x_i(t) - p_s^m(t) x_i(t),&\text{\small if $a_i(t)= s$,}\\
0,&\text{\small if $a_i(t) = n$,}
\end{cases}
\end{eqnarray}
where $\theta_i$ denotes a parameter of customer $i$'s willingness to pay for a service providing a certain level of QoS, e.g., $q_o^m(t)$ and $q_s^m(t).$ The parameter $\theta$ is assumed to be a random variable having a value within $[0,1]$, which follows a uniform distribution as popularly modeled in e.g.,~\cite{refAdoption}.  For example, if a customer has a lower willingness to pay, e.g., $\theta \rightarrow 0,$ then it might hardly get a positive utility from any kind of service, regardless of the service prices. Or, if a customer has a higher willingness to pay, e.g., $\theta \rightarrow 1,$ then it might have the incentive to use a high-end service and pay for it. Otherwise, if a customer has a mid-range willingness to pay, e.g., $\theta = 1/2,$ then its preference for service is hard to predict because it highly depends on price.
 Moreover, we note that the utility function is assumed to be non-negative, and it becomes zero if the customer $i$ decides not to use any service, i.e., $u_i(n) = 0$, as the baseline.

\smallskip
\noindent{\bf \em Revenue of cloud service provider.}
We denote by $\pi(\bm p(t))$ the CSP's total revenue\footnote{Note that the price vector $\bm p(t)$ is used for simplicity on notation, i.e., $\bm p(t) = [p_o^m(t), p_s^m(t)]_{m \in \set M}$.} at time slot $t,$ which is earned by the customers paying for on-demand or spot instance service. 
Let $\pi_o^m (p_o^m(t))$ and $\pi_s^m (p_s^m(t))$ be the revenues from the customers subscribing to on-demand and spot instance service in a cluster $m,$ respectively. Then, for an arbitrary cluster $m$ at time slot $t$, the CSP's revenue of on-demand or spot instance service is defined by service price, resource utilization, and market share as follows:
\begin{eqnarray}\label{eq:def_revenue}
\pi_a^m (p_a^m(t))
= p_a^m (t) \cdot \int_0^1 x_{\theta}(t) \cdot  \bm 1_{\{\theta \in \Theta_a^m(t)\}} d \theta,
\end{eqnarray}
for all $a \in \{o, s\},$ where $x_{\theta} (t)  = \sum_{i} x_i(t) \cdot \bm 1_{\{\theta_i = \theta \}}.$
We denote by $\Theta_a^m(t)$ the continuous interval of customers (or the market share for customers) subscribing to service $a \in \{o, s\},$ in cluster $m$ and at time slot $t.$
For example, at time slot $t,$ $\Theta_o^m(t) = [\frac{1}{2}, 1]$ implies that the customers having willingness to pay $\theta_i \in [\frac{1}{2}, 1]$ will choose to use the on-demand instance service.
Then, the total revenue of CSP, i.e.,  $\pi(\bm p(t))$, is given by the summation of all the revenues:
\begin{eqnarray}
\pi(\bm p (t)) = \sum_{m \in \set M} \big( \pi_o^m(p_o^m(t)) + \pi_s^m(p_s^m(t)) \big).
\end{eqnarray}



\smallskip
\noindent{\bf \em Game formulation.}
We formally describe the interaction between CSP and customers through the two-stage dynamic game. Again, we note that the main purpose of the game-theoretic analysis is to find the equilibrium, which is a stable point, from which no player has the incentive to deviate. This equilibrium is formed at the intersection of the best responses of players (i.e., CSP and customers). The best response of a player is defined by a set of strategies (i.e., price or service selection) maximizing profit (i.e., utility or revenue), each of which is regarding the other player's strategy.


\separator
\begin{compactenum}[$\circ$]
\item {\bf \em Stage I: Pricing of cloud service provider.} CSP decides service price vector $\bm p(t)$ for each time slot $t$, to maximize its revenue $\pi(\bm p (t))$:
\begin{align}\label{eq:max_revenue}
\bm p^\star (t) &= \arg \max_{\bm p(t)} \pi(\bm p(t)),\\ \nonumber
&= \arg \max_{\bm p(t)}  \sum_{m \in \set M} \Big(
p_o^m (t) \cdot \int_0^1 x_{\theta}(t) \cdot  \bm 1_{\{\theta \in \Theta_o^m(t)\}} d \theta \\ \nonumber
&+ p_s^m (t) \cdot \int_0^1 x_{\theta}(t) \cdot  \bm 1_{\{\theta \in \Theta_s^m(t)\}} d \theta
\Big ),
\end{align}
which is the best response of CSP regarding the customers' service selection, $\Theta_o^m(t)$ and $\Theta_s^m(t).$
\smallskip
\item {\bf \em Stage II: Customers' service selection.}
A customer under an arbitrary cluster region $m$ chooses one of the services $a_i(t)$ at time slot $t,$ for maximizing its utility $u_i(a_i(t))$:
\begin{align}\label{eq:max_utility}
    a^{\star}_i(t) &= \arg \max_{a_i(t) \in \set A} u_i(a_i(t); p_o^m(t), p_s^m (t)),~ \forall i \in \set I^m,
\end{align}
which is the best response of customer $i \in \set I^m,$ in regard to the CSP's price decision, $p_o^m(t)$ and $p_s^m(t)$.
\end{compactenum}
\separator

\section{Analysis of cloud service market: \\ on-demand and spot Instance Services}\label{sec:analysis}

As described in Section~\ref{subsec:game}, the two-stage dynamic game is formed over the assumption that CSP (i.e., leader) knows the behavior of customers (i.e., follower). To analyze this game, we use a standard tool, called backward induction, which is composed of two consecutive steps and results in the closed form of unique equilibrium.
As the first step, we analyze how customers select where to subscribe in Stage II for a given price vector $\bm p(t).$ It is followed by the second step analysis in Stage I, which is the price decision by CSP, who knows the analysis result of Stage II. 

\smallskip
\noindent{\bf \em Remark for notations.} We consider our game for a given cluster $m \in \set M,$ at a given time slot $t \in \set T,$ so that all the notations such as price, QoS, market share, utility, and revenue depend on $m$ and $t.$ However, for notational simplicity on analysis, we omit the superscription $m$ and the notation $t$, unless they are needed explicitly. For example, we use $p_o$ and $p_s$ instead of $p_o^m(t)$ and $p_s^m(t),$ respectively.

\subsection{Stage II: Customers' Service Selection}\label{subsec:stage2}
A customer $i$ under a region by cluster $m$, i.e., $i \in \set I^m,$ compares all the options of services, i.e., on-demand, spot, and no services, to maximize its utility $u_i.$ The following Lemma~\ref{lem:stage2} states the best response of customer $i$ having willingness to pay $\theta_i$ for a given price vector $\bm p.$
\begin{lemma}\label{lem:stage2}
At a time slot $t,$ a customer $i$ under a region by cluster $m$ selects an instance service $a_i \in \set A$ to maximize its utility as follows: 
\begin{eqnarray}
\label{eq:stage2}
a_i^{\star} = 
\begin{cases}
o, \quad \text{if} \quad \frac{p_o - p_s}{q_o - q_s} < \theta_i \leq 1,\\
s, \quad \text{if} \quad \frac{p_s}{q_s} < \theta_i \leq \frac{p_o - p_s}{q_o - q_s}, \\
n, \quad \text{otherwise.}
\end{cases}
\end{eqnarray}
Then, the market shares $\Theta_o$ and $\Theta_s$ can be defined by:
\begin{eqnarray}\label{eq:prop_marketshare}
\Theta_o = \Big(\frac{p_o- p_s}{q_o - q_s}, 1\Big] \;\; \text{and} \;\; \Theta_s = \Big(\frac{p_s}{q_s}, \frac{p_o-  p_s}{q_o - q_s}\Big].
\end{eqnarray}
Moreover, the spot instance service can have a positive market share if the following condition is met:
\begin{eqnarray}
{ {\tt C0:}} \;\; \frac{p_o}{q_o} > \frac{p_s}{q_s}.
\end{eqnarray}
\end{lemma}

\begin{proof}
Recall that a customer decides where to subscribe to maximize its utility. Following the definition of utility~\eqref{eq:utility}, the customers having $\theta_i \in [0,1]$ choose the on-demand instance service rather than the spot instance service if the following inequality holds:
\begin{gather*}
\theta_i q_o x_i - p_o x_i > \theta_i q_s x_i- p_s x_i, \cr
\theta_i > \frac{p_o - p_s}{q_o - q_s}.
\end{gather*}
It implies that the customers satisfying $\theta_i \in (\frac{p_o - p_s}{q_o - q_s}, 1]$ prefer to use the on-demand instance service. 
\begin{compactenum}[(a)]
    \item If the condition $\tt C0$ does not hold, the spot market share goes to zero because all the candidates of spot instance users have higher utilities when they use the on-demand instance service. This can be shown by which the following inequality is always true under the complement of condition $\tt C0$:
    \begin{gather*}
    \frac{p_s}{q_s} \geq \frac{p_o - p_s}{q_o - q_s}, \cr
    p_s q_o - p_s q_s \geq p_o q_s - p_s q_s, \cr
    \therefore p_s / q_s \geq p_o /q_o,
    \end{gather*}
        where the customers with $\theta_i \in (p_s / q_s, 1]$ satisfying $\theta_i p_s  - q_s > 0$ could be the candidate of spot instance users.
    \item Otherwise, the spot market share is positive because the customers having $\theta_i \in (\frac{p_s}{q_s}, \frac{p_o - p_s}{q_o - q_s}]$ will choose the spot instance service. The following inequality shows it is always true under the condition $\tt C0$:
    \begin{gather*}
        \frac{p_o - p_s}{q_o - q_s} > \frac{p_s}{q_s}, \cr
        p_o q_s - p_s q_s > p_s q_o - p_s q_s, \cr
        \therefore p_o / q_o > p_s / q_s.
    \end{gather*}
\end{compactenum}
Therefore, the condition $\tt C0$ should be satisfied for a positive spot market share. Subsequently, the market shares $\Theta_o$ and $\Theta_s$ can be defined as follows:
\begin{align*}
\Theta_o = \Big(\frac{p_o- p_s}{q_o -  q_s}, 1\Big] \;\; \text{and} \;\; \Theta_s = \Big(\frac{p_s}{q_s}, \frac{p_o- p_s}{q_o - q_s}\Big].
\end{align*}
\end{proof}

Lemma~\ref{lem:stage2} describes the service selection of customers having different willingness to pay $\theta_i$ for given prices $p_o$ and $p_s$. It also implies the on-demand instance service providing higher QoS attracts customers having a larger willingness to pay, and the spot instance service can have its market share in the middle range of willingness to pay.
Moreover, the condition $\tt C0$ represents that the unit price per unit quality of spot instance service, i.e., $p_s / q_s,$ should be lower than that of on-demand instance service, i.e., $p_o / q_o,$ for securing a positive amount of spot instance service's market share. 

\subsection{Stage I: Pricing of Cloud Service Provider}\label{subsec:stage1}

The service selection of customers in Lemma~\ref{lem:stage2} defines the market shares $\Theta_o$ and $\Theta_s$ for given prices $p_o$ and $p_s.$ We now study how the CSP decides the service price, which not only satisfies the equilibrium but also maximizes its revenue. Due to the complexity of the revenue function defined in~\eqref{eq:def_revenue}, it is challenging to find the closed-form of equilibrium analytically. Thus, we choose the model's simplicity for analytical tractability by assuming a uniform distribution of customers' willingness to pay $\theta_i$ as we explained in the definition of utility in~\eqref{eq:utility}. 
This assumption allows us to define the average resource utilization of each service in the following.

\smallskip
\begin{definition}\label{def:average resource utilization}
For each service $a \in \{o, s\},$ the average resource utilization of customers in $\Theta_a^m(t)$ is defined by $\gamma_a^m(t)$:
\begin{eqnarray}
    \gamma_a^m(t)  = \frac{\int_0^1 x_{\theta}(t) \cdot \bm 1_{\{\theta \in \Theta_a^m(t)\}} d \theta}{|\Theta_a^m(t)|},
\end{eqnarray}
where we recall $x_{\theta} (t) = \sum_{i} x_i(t) \cdot \bm 1_{\{\theta_i  = \theta\}}.$ Also, $|\Theta_a|$ denotes the length of interval $\Theta_a.$
\end{definition}




Thanks to the above definition, we can show the impact of customers' resource utilization on the CSP's decision explicitly in the following results. Note that we use $\gamma_o$ and $\gamma_s$ rather than $\gamma_o^m(t)$ and $\gamma_s^m(t)$ for notational simplicity.

\begin{lemma}\label{lem:condition}
At a time slot $t,$ a CSP has the unique equilibrium price for each service $a \in \{o, s\}$ in a cluster $m$, if the following condition $\tt C1$ is met:
\begin{eqnarray*}
{\tt C1}: \eta q_s / 2 q_o < \gamma_o < \eta /2,
\end{eqnarray*}
where we define $\eta=\gamma_o + \gamma_s$.     
\end{lemma}

\begin{figure*}[t!]
  \subfigbottomskip = 0cm
  \subfigcapmargin = 3pt
  \begin{center}
    \subfigure[\small Equilibrium prices
      \label{fig:price}]{
      \includegraphics[width=0.60\columnwidth]{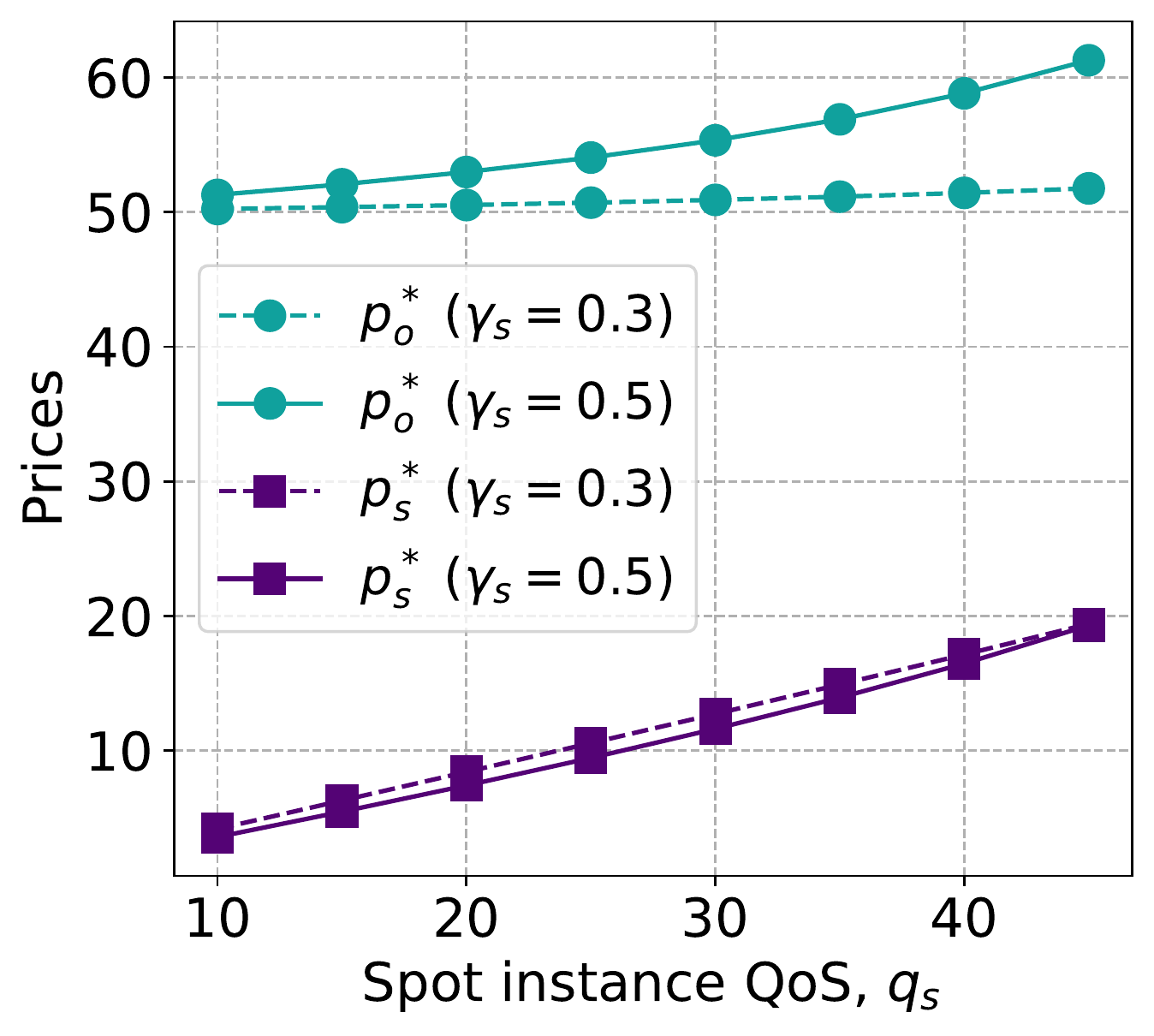}
    }  
    \subfigure[\small Equilibrium market shares
    \label{fig:market_share}]{
      \includegraphics[width=0.60\columnwidth]{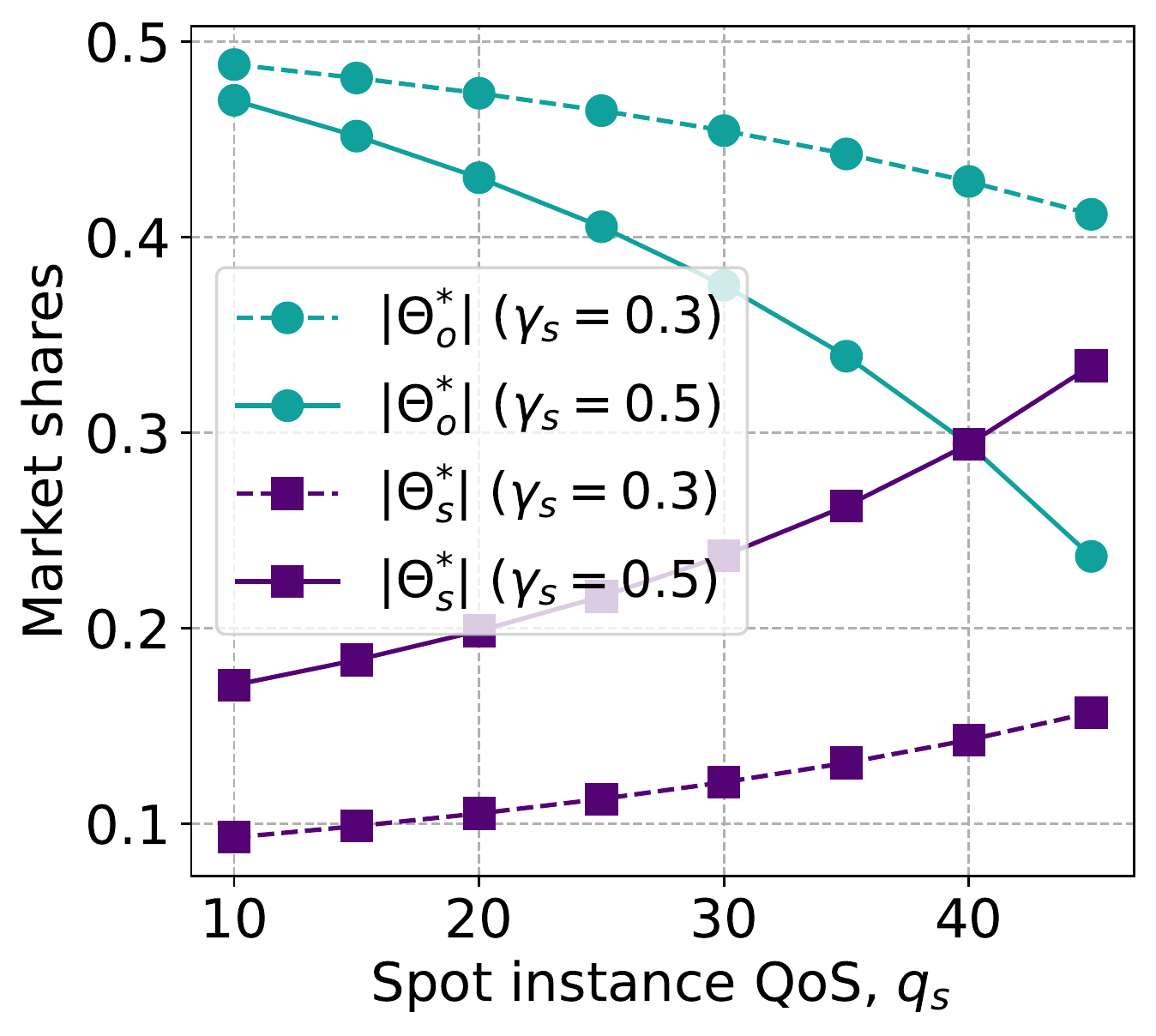}
    }
    \subfigure[\small Equilibrium revenues 
    \label{fig:revenue1}]{
      \includegraphics[width=0.60\columnwidth]{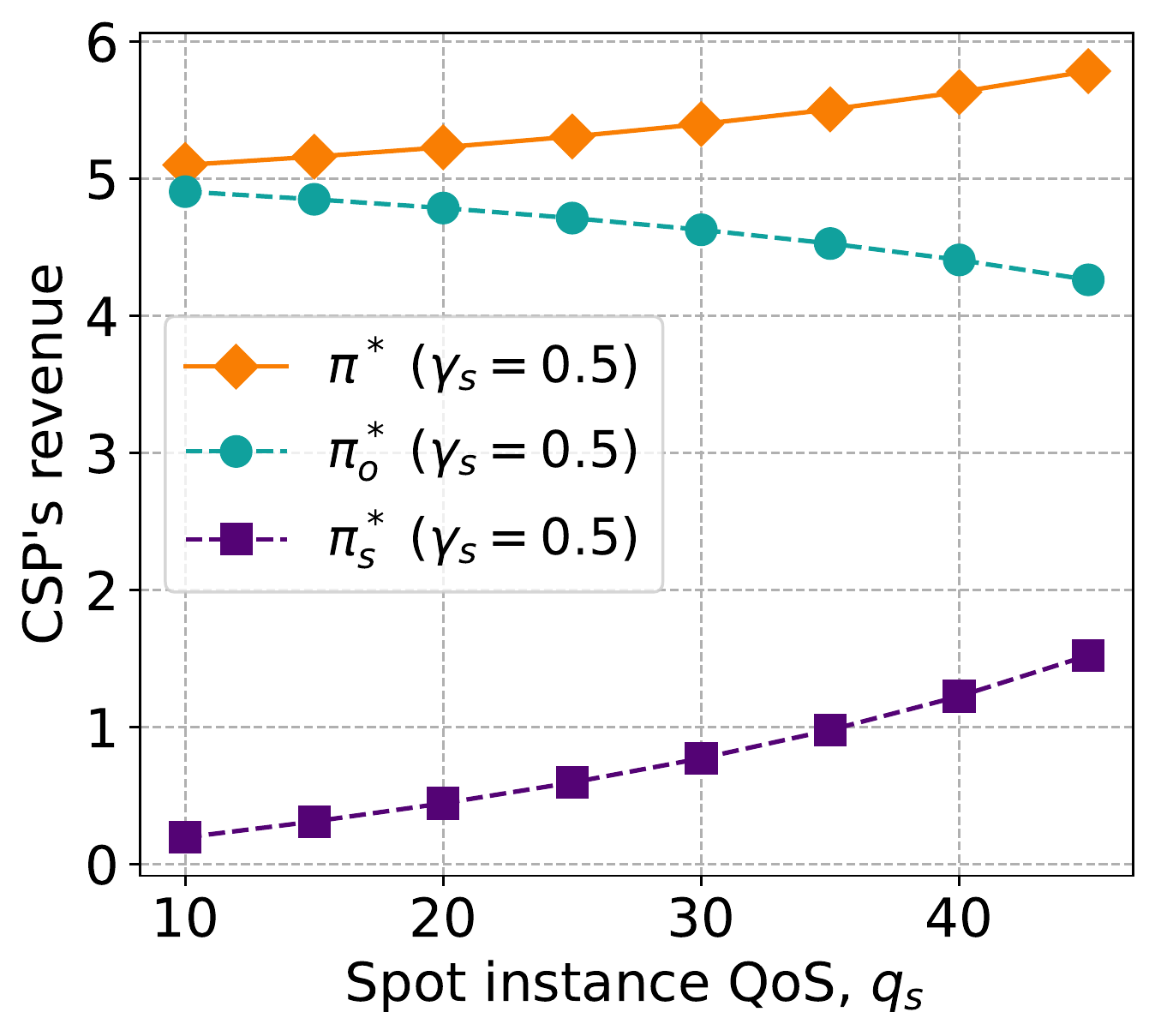}
    } 
  \end{center}
  \caption{Equilibrium prices, market shares, and revenues. We choose $q_o = 100,$ $\gamma_o = 0.2,$ $t = 1,$ and $m = 1$.}
  \label{fig:equilibrium for CSP}
\end{figure*}

\begin{proof}
We now provide a sketch of the proof. The full proof is in Appendix~\ref{proof:lem_condition}.
By backward induction, we get the revenue $\pi(\bm p)$ as the function of pricing strategy $\bm p$ using the market share derived in Lemma~\ref{lem:stage2}. Then, under $\tt C1,$ we show that $\pi (\bm p)$ is a strictly concave function with respect to the price $\bm p$ because the Hessian of $\pi(\bm p)$ is negative definite. Therefore, there exists the unique equilibrium price $\bm p^\star$.
\end{proof}

The condition $\tt C1$ provides useful implications to CSP for understanding the viability of spot instance service. 
Note that the mildness of $\tt C1$ is critical to guarantee the viability of spot instance service because we cannot get the unique equilibrium price of both services if the condition $\tt C1$ does not hold.
The condition $\tt C1$ explicitly describes the range of resource utilization of each instance service.
On the left-hand side of $\tt C1$, the lower bound of $\gamma_o$ goes to zero if the QoS of on-demand instance service $q_o$ is extremely larger than that of spot instance service $q_s.$
It implies that the spot service is easy to profitable if the difference of QoSes between on-demand and spot is large enough.
On the right-hand side of $\tt C1$, the upper bound of $\gamma_o$ is half of the total utilization $\eta$, which means that the resource utilization of on-demand service should be lower than half of capacity, i.e., $c/2,$ to find a profitable $\gamma_s.$

Theorem~\ref{thm:equilibrium} states the equilibrium service price, which is the stable strategic decision of CSP so as to maximize its revenue.
Furthermore, the equilibrium market share of each service and the total revenue of CSP are also derived from the equilibrium price.

\separator

\begin{theorem}\label{thm:equilibrium}
Under the condition $\tt C1$, a CSP has the unique closed-form equilibrium price of each instance service:
\begin{align} \label{eq:eq_price}
    p_o^{\star} = \frac{2\gamma_o \gamma_s q_o (q_o - q_s)}{4\gamma_o \gamma_s q_o - \eta^2 q_s}, ~p_s^{\star} = \frac{\eta \gamma_o q_s (q_o - q_s)}{4\gamma_o \gamma_s q_o - \eta^2 q_s}.
\end{align}
Then, the market share of each service is derived as follows:
\begin{align} \label{eq:eq_marketshare}
|\Theta_o^\star| = \frac{\gamma_s(2\gamma_o q_o - \eta q_s)}{4 \gamma_o \gamma_s q_o - \eta^2 q_s}, ~|\Theta_s^\star| = \frac{\gamma_o(\gamma_s - \gamma_o)q_o}{4\gamma_o \gamma_s q_o - \eta^2 q_s},
\end{align}
which leads to the total revenue of CSP as follows:
\begin{align} \label{eq:eq_revenue}
\pi^{\star} = \frac{\gamma_o^2 \gamma_s q_o(q_o - q_s)}{4\gamma_o \gamma_s q_o - \eta^2 q_s}.
\end{align}
Recall that for notational simplicity, we omitted $t$ and $m,$ e.g, $p_o^{\star}$ and $\Theta_o^{\star}$ refer to ${p_o^m}^\star(t)$ and ${\Theta_o^m}^\star(t),$ respectively. 
\end{theorem}

\separator



The proof of Theorem~\ref{thm:equilibrium} is presented in Appendix~\ref{proof:thm}. We obtain the unique closed-form equilibrium, which represents the stable pricing strategy according to rational customers with heterogeneous preferences for QoS at a specific time slot. From the equilibrium analysis, we can derive two types of interpretations. The first one, in Section~\ref{subsec:numerical analysis of equilibrium}, is the numerical analysis to understand the impacts of static equilibrium itself, and the second one, in Section~\ref{sec:evaluation}, is the experimental evaluation demonstrating how the equilibrium-based pricing strategy can operate in the real-time cloud market, along with the instance provisioning algorithms.

\section{Numerical Examples: \\Interpretation of Equilibrium}\label{subsec:numerical analysis of equilibrium}

We now show numerical examples to understand the equilibrium: how the QoS (i.e., $q_o$ and $q_s$) and the customers' resource utilization (i.e., $\gamma_o$ and $\gamma_s$) impact on price, market share, and revenue. We consider a setup within the condition $\tt C1$. 
We first choose a homogeneous QoS for on-demand instance service with $q_o = 100$ and a fixed average resource utilization of on-demand instance service with $\gamma_o = 0.2.$ Meanwhile, we set a larger average resource utilization of spot instance service, i.e., $\gamma_s = 0.3$ or $0.5,$ and a varying QoS of spot instance service with $q_s = 10$ to $50.$ We choose other parameters as $t = 1$ and $m = 1$, which stand for a single time slot and cluster.



\subsection{Impacts of Spot Instance Service on CSP}

\smallskip
\noindent{\bf \em Impacts on on-demand instance service price and market share.}
Interestingly, the on-demand instance service price $p_o^\star$ increases as the spot instance's QoS $q_s$ as shown in Fig.~\ref{fig:price}. This phenomenon can be jointly explained with the on-demand market share $\Theta_o^\star$ in Fig.~\ref{fig:market_share}, which decreases with respect to $q_s.$ This implies the customers with a higher willingness to pay $\theta$ remain at the on-demand instance service while customers having a lower willingness to pay $\theta$ move to a cheaper service (i.e., spot instance service). Thus, the CSP can be paid a higher price by high-end customers. This tendency is observed more clearly when spot instance's resource utilization increases from $\gamma_s = 0.3$ to $\gamma_s = 0.5.$

\smallskip
\noindent{\bf \em Impacts on spot instance service price and market share.}
Fig.~\ref{fig:price} also illustrates that the spot instance price $p^\star_s$ increases as spot instance's QoS $q_s$. This implies that CSP can be paid a higher price $p_s$ when it can guarantee improved QoS for spot instance service. It is not a straightforward but interesting result because it implies the spot instance service does not cause perfect competition, in which $p^\star_o$ and $p^\star_s$ go to zero, against the on-demand instance service. In general, the perfect competition in pricing frequently happens among services sharing potential customers, and it tends to be severe as the difference between the services' QoSes decreases.
However, we do not have this phenomenon because the spot instances only use the remaining resource from the on-demand instance service, thus it does not threaten the existing instance services in terms of pricing.
Moreover, the market share of spot instance $\Theta_s^\star$ is also enlarged with its QoS $q_s$ as shown in Fig.~\ref{fig:market_share}. This implies that the spot instance service takes the customers having a mid-range willingness to pay from the on-demand instance service. 

\smallskip
\noindent{\bf \em Impacts on CSP's revenue.}
Fig.~\ref{fig:revenue1} shows that the total revenue $\pi^\star = \pi^\star_o + \pi^\star_s$ increases with spot QoS $q_s$. This is due to the fact that the revenue increase of the spot instance service $\pi_s^\star$ can compensate for the diminishing revenue of the on-demand instance service $\pi_o^\star$. The revenue from on-demand instance service $\pi_o^\star$ is slightly decreasing with $q_s$ because the equilibrium price $p_o^\star$ increases, but it loses its market share $\Theta_o^\star$ as shown in Fig.~\ref{fig:price} and Fig.~\ref{fig:market_share}.

\subsection{Comparison to On-demand-only Instance Service}

The advantage of the spot instance service is also revealed by comparison with the case where the CSP only serves the on-demand instance service on its clusters. Lemma~\ref{lem:on-demand-only} states the result of equilibrium analysis for the market having only on-demand instance service, which is through the same process (i.e., backward induction) as in Section~\ref{subsec:stage2} and~\ref{subsec:stage1}.

\begin{figure}[h]
  \centering
  \includegraphics[width=0.57\columnwidth]{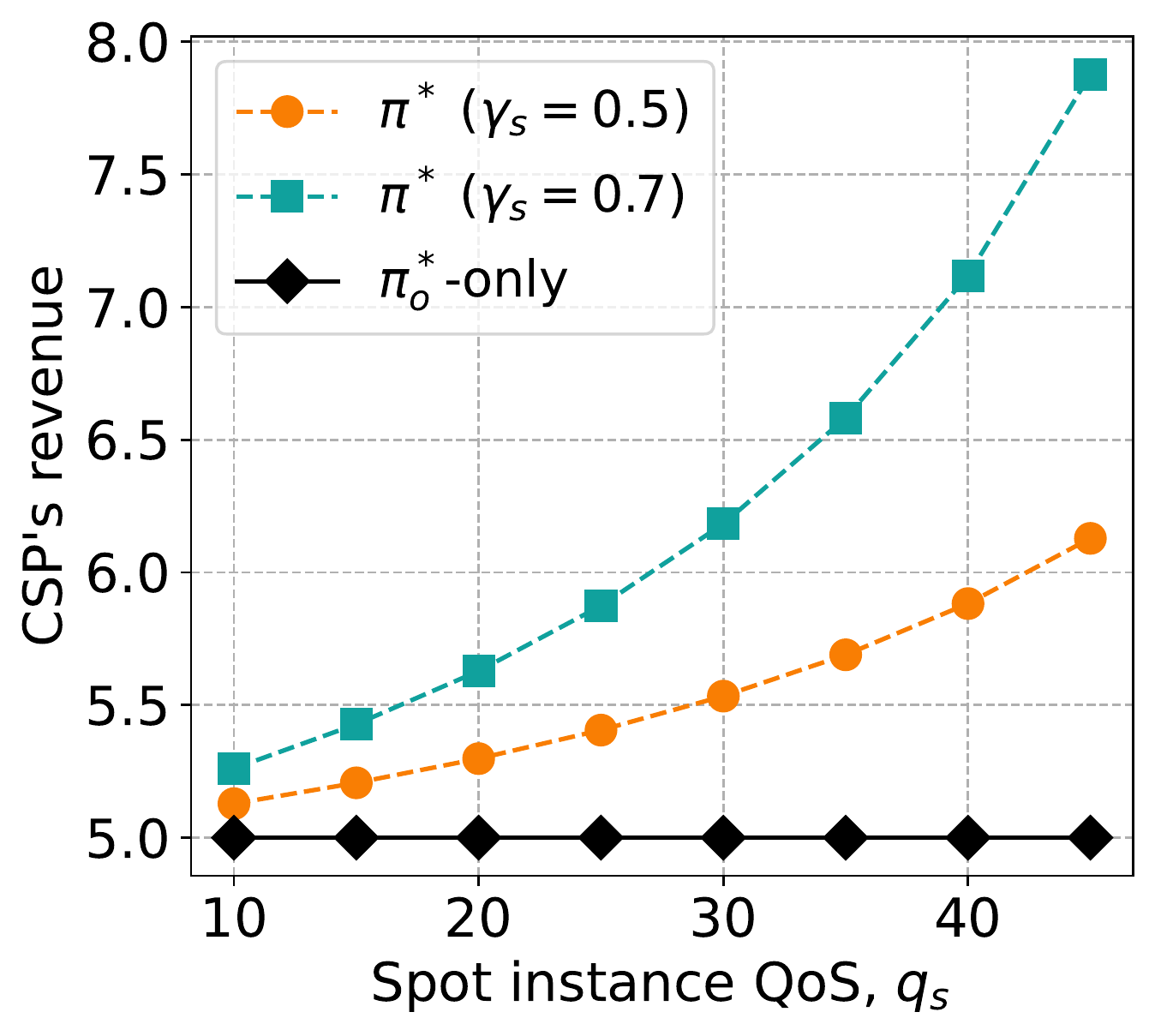}
  \vspace{-0.2cm}
  \caption{\small CSP's revenue at the equilibrium when (i) on-demand only case ($\pi_o^\star$-only) and (ii) spot + on-demand case ($\pi^\star$). We choose the parameters: $q_o = 100$, $\gamma_o = 0.2$, $t = 1,$ and $m = 1.$ }
  \label{fig:revenue_comparison}
\end{figure}

\begin{lemma}\label{lem:on-demand-only}
   If a CSP only serves on-demand instance service, then at the equilibrium, the service price, the market share, and the revenue of CSP are determined as:
\begin{align}
\bar p_o^\star = \frac{q_o}{2}, \;\; |\bar \Theta_o^\star| = \frac{1}{2}, \;\; \text{and} \;\; \bar \pi^\star = \frac{\gamma_o q_o}{4},
\end{align}
where to avoid confusion, we use a bar in the notation.
\end{lemma}

The proof is described in Appendix~\ref{proof:lem:on-demand-only}.
Note that we omitted $t$ and $m$ in notations for simplicity. Lemma~\ref{lem:on-demand-only} describes that in the market where only on-demand instances are served, the equilibrium price is set as half of QoS. At this time, half of the customers use the service, and the other does not use it. Then, the CSP's revenue is determined by multiplying the price and the market share. By comparing Lemma~\ref{lem:on-demand-only} to Theorem~\ref{thm:equilibrium}, we can interpretively understand how CSPs should change their pricing policy by adopting spot instance service. Accordingly, we also know how their market share changes and how much their revenue increase.

\begin{table*}[t]
\begin{center}
\renewcommand{\tabcolsep}{3mm}
  \renewcommand{\arraystretch}{1.7} 
\begin{tabular}{@{}cccc@{}}
\toprule
Market   & No-service ($\Theta_n^\star$) & spot ($\Theta_s^\star$) & on-demand ($\Theta_o^\star$)     \\ \midrule
on-demand only &  $[0,\frac{1}{2})$ & $\emptyset$ & $ [\frac{1}{2},1]$  \\ \midrule
on-demand + spot &  $\Big[0,\frac{\eta \gamma_o (q_o - q_s)}{4\gamma_o \gamma_s q_o - \eta^2 q_s}$ \Big)  & $\Big[\frac{\eta \gamma_o (q_o - q_s)}{4\gamma_o \gamma_s q_o - \eta^2 q_s}, \frac{2\gamma_o \gamma_s q_o - \eta \gamma_o q_s}{4\gamma_o \gamma_s q_o - \eta^2 q_s} \Big)$ & $ \Big[\frac{2\gamma_o \gamma_s q_o - \eta \gamma_o q_s}{4\gamma_o \gamma_s q_o - \eta^2 q_s} , 1\Big]$ \\ \bottomrule
\end{tabular}
\vspace{0.3cm}
\caption{Market share at the equilibrium}
\label{table}
\end{center}
\vspace{-0.7cm}
\end{table*}

\begin{proposition}\label{prop:monopoly}
Under the condition $\tt C1,$ we get the following properties by comparison of two markets, i.e., on-demand only versus on-demand + spot:
\begin{compactenum}[(a)]
    \item (price) $p_o^\star  > \bar p_o^\star$,
    \item (market share) $|\Theta_o^\star| < |\bar \Theta_o^\star|$ and
    $|\Theta_o^\star| + |\Theta_s^\star| > |\bar\Theta_o^\star|$,
    \item (revenue) $\pi^\star > \bar \pi^\star$,
\end{compactenum} 
where we omit $t$ and $m$ in above notations for simplicity.
\end{proposition}

\smallskip
The proof is shown in Appendix~\ref{proof:prop:monopoly}.
Proposition~\ref{prop:monopoly} implies that CSP gets a higher price and larger revenue when it both serves on-demand and spot instance services, compared to the case where the only on-demand instance is being served. 
In Fig.~\ref{fig:revenue_comparison}, we can clearly see how much more revenue CSP can earn by serving spot instances alongside on-demand instances. This tendency becomes stronger as the quality of spot instance service ($q_s$) improves and the user's utilization ($\gamma_s$) increases.
This is non-trivial in the market because there exists a price competition between services for market share, and it may cause to decrease in providers' equilibrium price and revenue~\cite{arrow1954existence}. In our analysis, however, we show that the CSP can increase its revenue by offering spot instance service. This is caused by the fact that the spot instance service enlarges the aggregate market share, and it does not induce severe cannibalization among services by using idle resources.

\subsection{Impacts of Spot Instance Service on Customers}

So far, our analysis has focused on the implication for the CSP. In this section, we aim to understand the impact of spot instance service on customers' service selection and from which turns in its utility. To this end, we apply the equilibrium price in Theorem~\ref{thm:equilibrium} to customers' best response for service selection derived in Lemma~\ref{lem:stage2}. The result is in Table~\ref{table}, which states the market share $\Theta_a^\star$ for each service $a \in \{o, s\}$. Again, we omit $t$ and $m$ for notational simplicity.


\begin{figure}[h]
  \subfigbottomskip = 0cm
  \subfigcapmargin = 3pt
  \begin{center}
    \subfigure[\small Customers' utility per unit resource usage ($q_s = 30$)
      \label{fig:customer_utility}]{
      \includegraphics[width=0.47\columnwidth]{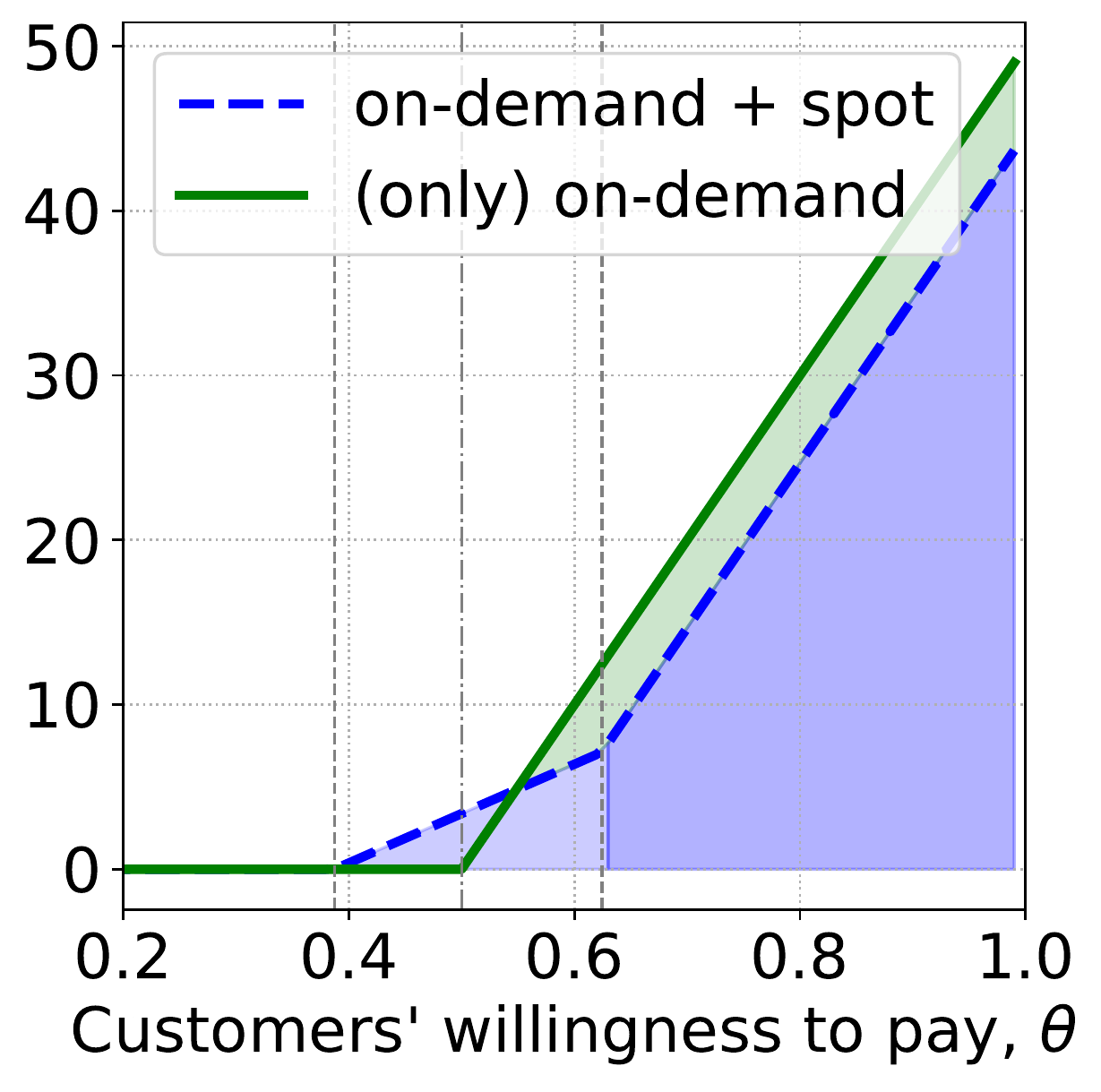}
    }  
    \hspace{-0.3cm}
    \subfigure[\small Aggregate customers' utility
    \label{fig:customer_aggregate_utility}]{
      \includegraphics[width=0.47\columnwidth]{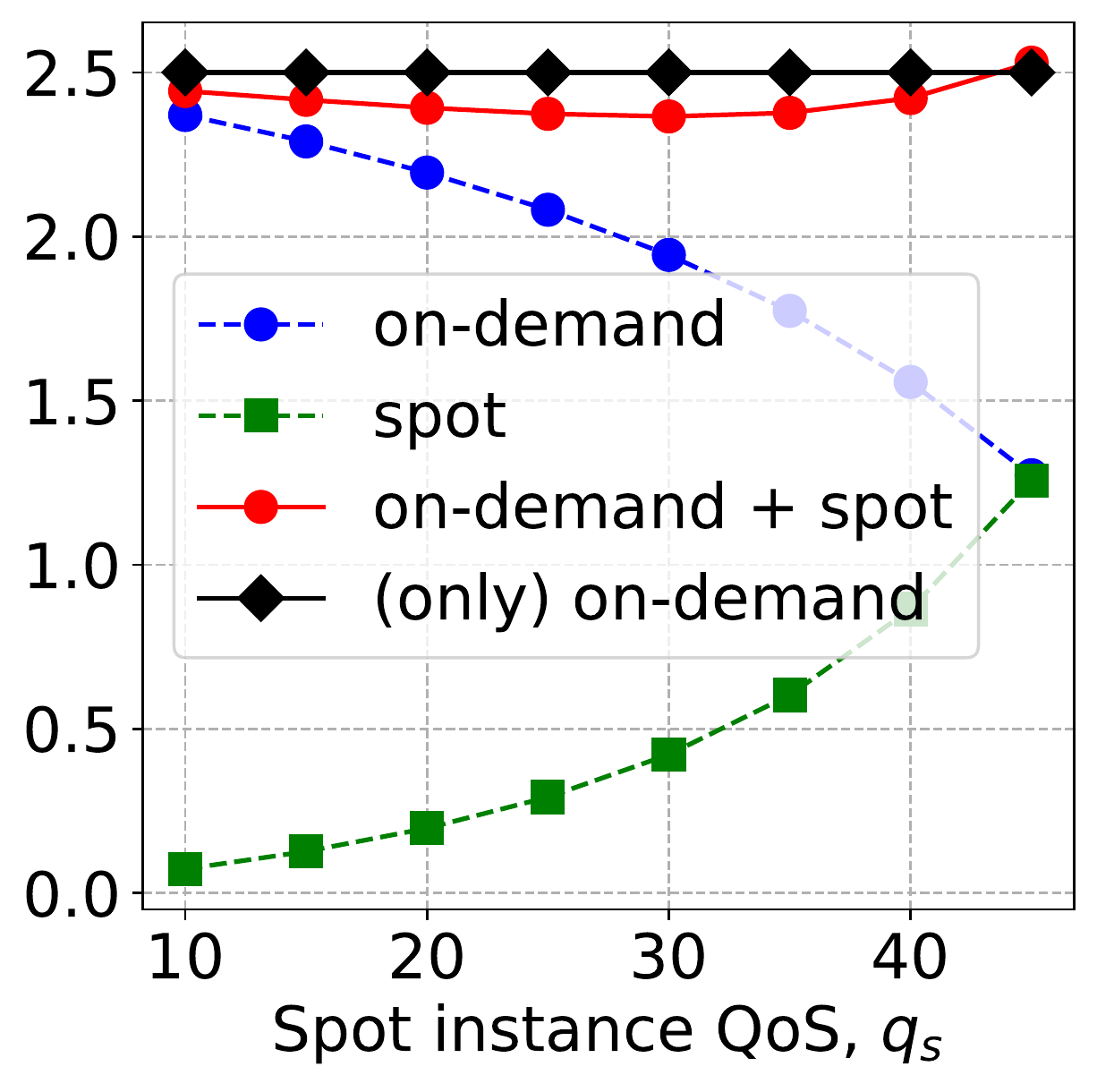}
    }
  \end{center}
  \caption{Customers' utility at the equilibrium. We choose the parameters: $q_o = 100,$ $\gamma_o = 0.2,$ $\gamma_s = 0.5,$ $t = 1,$ and $m = 1$.}
  \label{fig:uility}
\end{figure}

Fig.~\ref{fig:customer_utility} shows both the service selection and the utility per unit resource usage, i.e., $u_i(\theta) / x_i,$ according to the customers' willingness to pay $\theta.$ When CSP provides both on-demand and spot services, illustrated with the blue dashed line in Fig.~\ref{fig:customer_utility}, the users with $\theta \in [0.41, 0.62)$ choose the spot instance, and the users with $\theta \in [0.62, 1]$ choose the on-demand instance.
In addition, when CSP only offers on-demand instances, described as the green solid line in Fig.~\ref{fig:customer_utility}, the customers with $\theta \in [0.5,1]$ subscribe to the on-demand service.
Therefore, by comparing the two graphs, we can see that adopting a spot instance service can reduce the customers' utility per unit resource while increasing CSP's total market share.

\begin{figure*}[t!]
  \subfigbottomskip = 0cm
  \subfigcapmargin = 3pt
  \begin{center}
    \subfigure[\small Average CPU utilization
    \label{fig:cpu_utilization}]{
      \includegraphics[width=0.55\columnwidth]{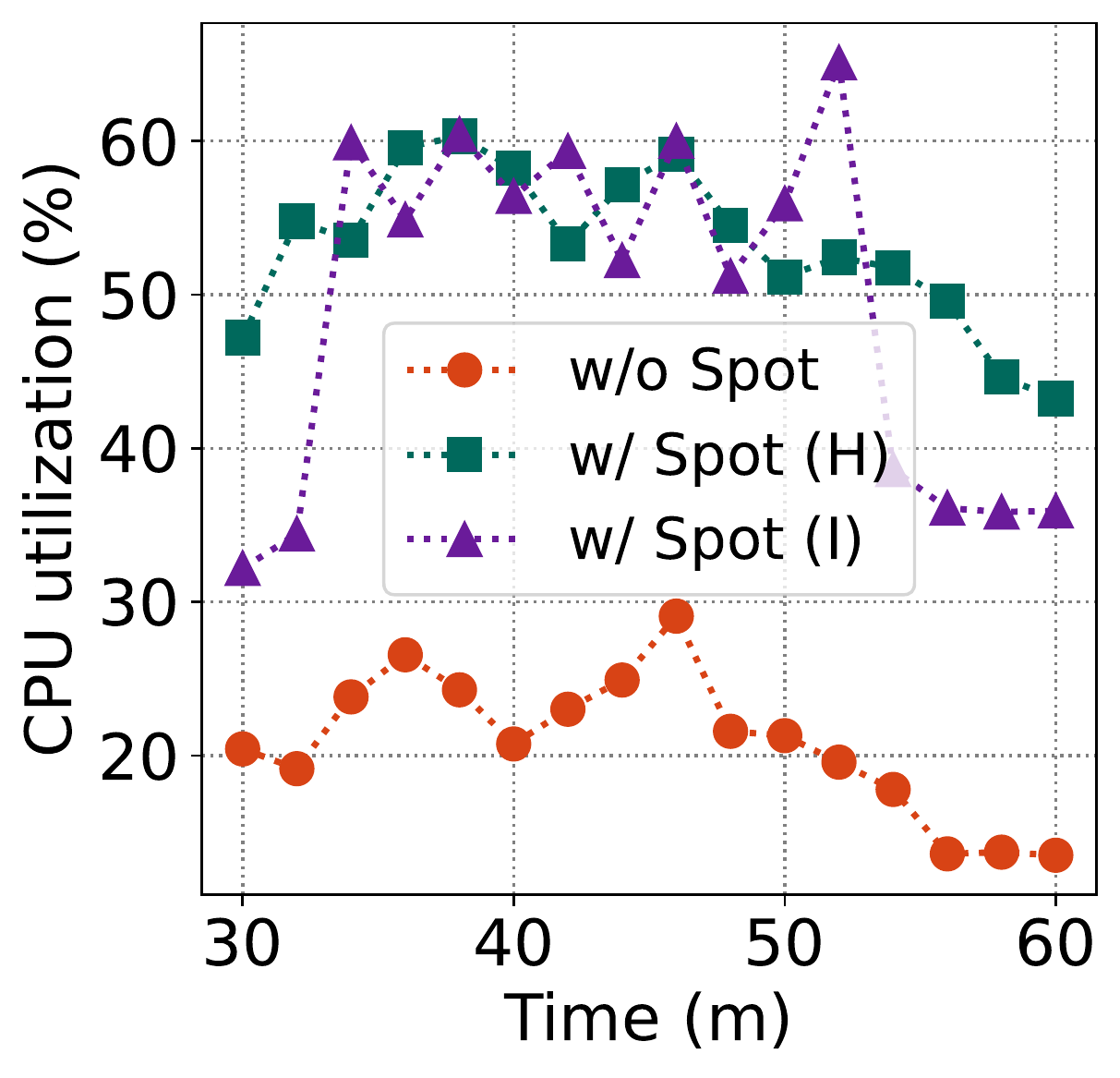}
    }\hspace{0.1cm} 
    \subfigure[\small Average RAM utilization
    \label{fig:ram_utilization}]{
      \includegraphics[width=0.55\columnwidth]{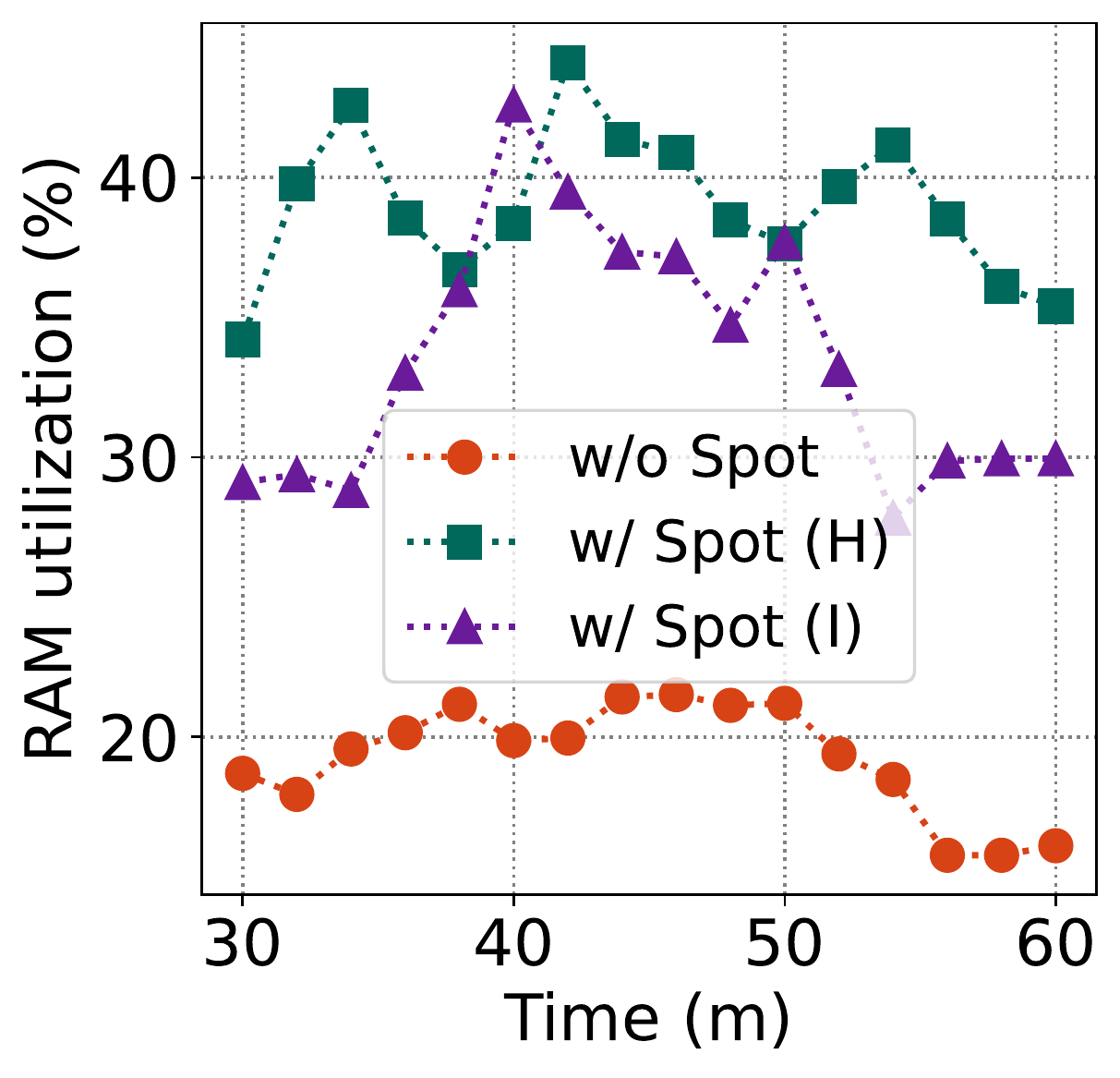}
    }\hspace{0.1cm} 
    \subfigure[\small Revenue
    \label{fig:revenue}]{
      \includegraphics[width=0.55\columnwidth]{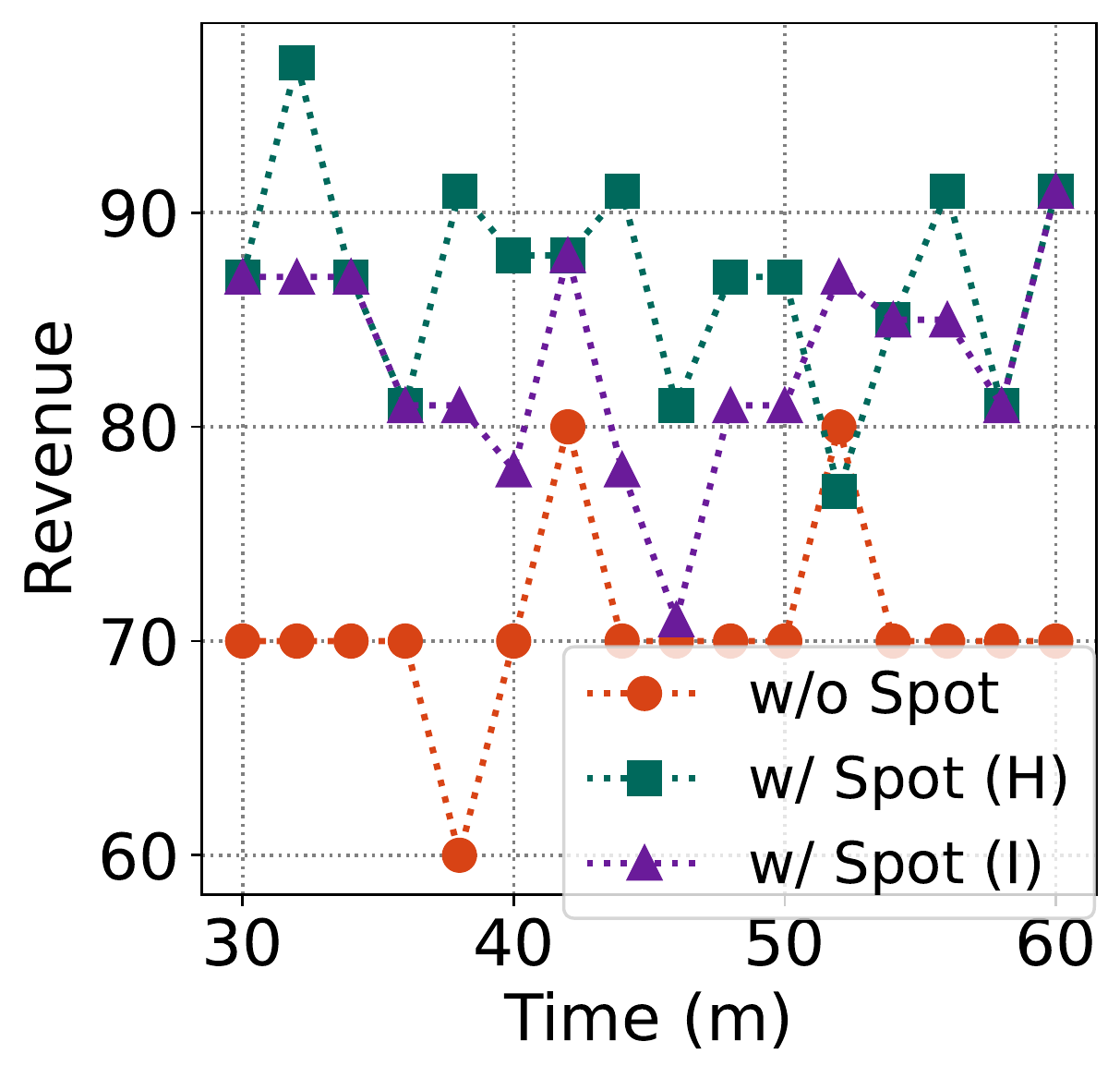}
    }
  \end{center}
  \caption{Comparison of on-demand only case (w/o spot), on-demand + spot with heuristic algorithm case (w/ spot (H)), and on-demand + spot with ILP algorithm case (w/ spot (I)). 
  }
  \label{fig:fig3} 
\end{figure*}

\begin{proposition}\label{prop:customer}
Under the condition $\tt C1,$ we obtain the aggregate customers' utility for each service at the equilibrium:
\begin{compactitem}
\item {on-demand instance service:}
\begin{align*}
\int_{\Theta_o} u_i^\star(\theta) &~ d \theta =
     q_o \gamma_o \gamma_s (2 \gamma_o q_o - \eta q_s) \\
     \times & \Big(\frac{\gamma_o q_s(\gamma_s - \gamma_o) + 2 \gamma_o \gamma_s (q_o + q_s) - \eta^2 q_s}{2(4\gamma_o \gamma_s q_o - \eta^2 q_s)^2}\Big),
\end{align*}
\item {spot instance service:}
\begin{align*}
    \int_{\Theta_s} u_i^\star(\theta)~ d \theta = \frac{\gamma_o^2 \gamma_s q_o^2 q_s (\gamma_s- \gamma_o)^2}{2(4\gamma_o \gamma_s q_o - \eta^2 q_s)^2}.
\end{align*}
\end{compactitem}
\end{proposition}

\smallskip
The proof is described in Appendix~\ref{proof:prop:customer}. Proposition~\ref{prop:customer} states the aggregate utility of customers for each service when CSP serves both on-demand and spot instances. In Fig.~\ref{fig:customer_aggregate_utility}, we observe that the aggregate utility of on-demand service (the blue dashed line with circles) decreases with $q_s.$ This is because the market share of on-demand service decreases, and the price increases as $q_s$ increases, as shown in Fig.~\ref{fig:equilibrium for CSP}. On the other hand, the aggregate utility of spot instances (the green dashed line with squares) increases because both its price and market share increase, as also seen in~Fig.\ref{fig:equilibrium for CSP}. Finally, the sum of the aggregate utilities (the red solid line with circles) tends to be smaller but can be greater than the case of providing only the on-demand instance service (the black solid line with diamonds), depending on the QoS of the spot instance and the amount of resource utilization. In other words, servicing spot instances can increase not only the CSP's revenue but also the customer's total surplus under some conditions.

\section{Experimental Evaluation}\label{sec:evaluation}

\subsection{Environment}

In this section, we experiment the spot instance provisioning with pricing.
To this end, a simple test bed based on OpenStack is constructed, composed of a single controller node and three compute nodes (i.e., three computing clusters). For experimental tractability, we use a deterministic instance request trace, a mixture of spot and on-demand instance requests. We decide which compute node to provision a new instance by (a) heuristic algorithm or (b) integer linear programming (ILP)~\cite{vanderbei2020linear} algorithm for securing instances' priority and load-balancing among compute nodes.
Recall that according to SLA, the on-demand instance has a higher priority on provisioning than the spot instance. 
Each provisioned instance is assumed to generate workloads randomly that follow the Poisson distribution.
Moreover, the spot price is determined by two marginal prices, one of which is the lower bound of price, which is given by the analysis result in Section~\ref{sec:analysis}, and another is the upper bound of price based on customers' maximum bids similar to multi-unit auction mechanism~\cite{krishna2009auction}.

\subsection{Spot and On-demand Instance Provisioning Algorithms}

Two simple provisioning algorithms are deployed to understand how resource utilization and the CSP's revenue are affected by the spot/on-demand instance provisioning strategy. For both algorithms, we use two thresholds of computing resource (i.e., CPU and RAM) utilization for managing resources: $th_{soft}$ and $th_{hard},$ where $th_{soft} < th_{hard}.$ For example, if a compute node's resource utilization is lower than $th_{soft},$ then all spot instances are safe, and a new spot request can also be provisioned. 
However, some of the spot instances can be reclaimed by CSP if the node's resource utilization is higher than $th_{hard}$ until the utilization of the node goes under the hard threshold. 
Furthermore, if a new on-demand instance request is submitted and the node's utilization is over the soft threshold, then
we terminate the spot instances until under the soft threshold and provision the on-demand instance.


\smallskip
\begin{compactenum}[\em (a)]
\item {\em Heuristic algorithm:} Traversing all nodes for finding a target compute node for provisioning spot/on-demand instance and terminating spot instance.
\item {\em ILP algorithm:} Formulating a problem in which the provisioning/terminating decisions are set as the integer decision variables. The thresholds are reflected as the constraints of the problem. OR-Tools 6.10~\cite{ORTools} is used to solve this integer linear programming.
\end{compactenum}



\subsection{Impacts of Spot on Resource Utilization and Revenue} 
 
 Figure~\ref{fig:fig3} illustrates a comparison of the cases, which are
 on-demand only (w/o spot), on-demand + spot using the heuristic algorithm (w/ spot (H)), and on-demand + spot using ILP algorithm (w/ Spot (I)), in terms of CPU/RAM utilization and CSP's revenue, which are the average of three computing clusters at each time slot. Fig.~\ref{fig:cpu_utilization} and Fig.~\ref{fig:ram_utilization} respectively depict the average CPU and RAM utilization where the thresholds are assumed to be $th_{soft} = 50\%$ and $th_{hard} = 70\%.$ In Fig.~\ref{fig:revenue}, we set the on-demand price as $10$ and the lower bound of spot price auction as $3$, which reflects the analysis result in Section~\ref{sec:analysis}. Then, the findings from Figure~\ref{fig:fig3} can be summarized in what follows:
 \begin{compactitem}[$\circ$]
 \item When the resource utilization by on-demand instances is sufficiently low, as seen in Fig.~\ref{fig:cpu_utilization} and Fig.~\ref{fig:ram_utilization}, by adopting spot instances in the cloud, CSP can achieve the revenue increase as shown in Fig.~\ref{fig:revenue}.
 \item Under both heuristic and ILP algorithms, the CPU utilization is well-managed under the hard threshold, i.e., $th_{hard} = 70 \%,$ as shown in Fig.~\ref{fig:cpu_utilization}. 
 \item In Fig.~\ref{fig:revenue}, at the time epoch of $52$ minute, the green square marker (when the Spot instances are provisioned by the heuristic algorithm) is below the red circle marker (without Spot instances). It indicates the possibility of revenue loss when we adopt the spot instance service with an improper provisioning algorithm. 
 This is because the lower price spot instance hinders the higher price on-demand instance provisioning. To avoid this phenomenon, we might need a proper eviction and provisioning strategy of spot instances, which would be included in our future direction.
 \end{compactitem}

 


\section{Related Work}

Cloud service providers consider Spot instances an attractive solution that exploits the underutilized resources in their cloud and makes an additional profit. Since AWS started providing spot instances in 2009~\cite{AWSSpot}, Preemptible VMs of Google Cloud~\cite{GCPSpot}, Spot VMs of Azure~\cite{AzureSpot}, and Transient Virtual Servers of IBM Cloud~\cite{IBMSpot} have been operating Spot instance services in a similar form with AWS. In addition, an auction-like algorithm adopting maximum bids submitted by customers has been used for pricing the Spot instance services in~\cite{AWSSpot, AzureSpot, IBMSpot, GCPSpot}, because the auction is an efficient tool for solving the resource allocation and the pricing decision concurrently. In 2017, AWS announced that it no longer uses auction directly but determines the Spot price according to the resource supply and the demand of Spot instances~\cite{AWSSpotArticle}. However, they still use an auction-like Spot pricing in which the customers' maximum bids are used as the upper bound of price, so that if the Spot price goes up above the bids, then Spot is reclaimed.


Including the paper in~\cite{agmon2013deconstructing}, there are many papers that have investigated how to determine the Spot price in the cloud resource market, as seen in surveys~\cite{luong2017resource, kumar2018survey, cong2020survey}. In~\cite{agmon2013deconstructing}, the authors deconstructed the AWS Spot price, so that they found out Spot price might not be originated from a pure market-driven auction but was artificially set by AWS. Moreover, the papers in~\cite{toosi2016auction, li2018online, shi2014online,  hosseinalipour2021two, zheng2015bid, khodak2018learning,  song2020pricing}
have analyzed the economic impacts of Spot price based on an auction-like mechanism. In~\cite{toosi2016auction}, the authors designed an auction mechanism for cloud spot markets to maximize the cloud provider's profit. In~\cite{li2018online}, the authors proposed an online auction mechanism for maximizing the revenue of IaaS clouds, which is composed of the allocation rule for resource sharing and the payment rule. The authors in~\cite{shi2014online} designed the online combinatorial auction for maximizing the social welfare of the VM market in cloud computing. Also,~\cite{hosseinalipour2021two} suggested a two-stage auction mechanism that captures the interaction among customers, cloud managers, and cloud providers.
The authors in~\cite{zheng2015bid, khodak2018learning} studied the bidding strategy of customers to minimize the cost of completing a job by predicting the spot price.
 The paper~\cite{song2020pricing} suggested pricing and bidding strategies for cloud providers and users of Spot instance services considering the delays for job completion.

Although our paper does not directly deal with the auction mechanism, it provides the guideline for Spot price, which is derived from the relationship between supply and demand of Spot instances when the CSPs set the Spot price to decide the Spot users based on their maximum bids. To figure out the providers' pricing strategy from game-theoretic modeling and equilibrium analysis is common in the communication network area~\cite{sen2010modeling, lee2018incentivizing, kim2019economics}. Similarly, the papers in~\cite{ma2012pricing, xu2013dynamic, alzhouri2017dynamic, alzhouri2021maximizing, li2021profits, di2013optimal, cardellini2016game} studied the pricing of the cloud market via game-theoretic analysis. In~\cite{ma2012pricing}, the best job submission strategy of customers was considered, where the dynamic changing spot price is modeled by a multi-stage game. In~\cite{xu2013dynamic}, the authors represented the dynamic pricing problem for maximizing revenue.
The papers~\cite{alzhouri2017dynamic, alzhouri2021maximizing} also investigated the dynamic pricing strategy by stochastic Markov decision process to maximize cloud revenue. In~\cite{li2021profits}, the author studied the competitive interaction between CSPs in a non-cooperative game theoretic model. The papers in~\cite{di2013optimal, cardellini2016game} analyzed the providers' strategies by a two-stage dynamic game between IaaS and SaaS providers.
 To the best of our knowledge, our paper is the first that studies the equilibrium analysis with a game-theoretic approach, which captures the resource utilization of both On-demand and Spot instances of CSPs with respect to the customers' service selection according to their heterogeneous willingness to pay on the QoS.

 \section{Future work}
We can improve the system model, in particular, customers’ heterogeneity in willingness to pay on the different QoS levels, which is currently a uniformly random value. Some customers may pay a premium price to get a higher level of QoS if their job is mission-critical, while others are more focused on cost-saving. Therefore, a deeper understanding of their jobs and required QoS levels would improve a future system model. 
Furthermore, our model can be upgraded by considering the dynamics of resource utilization according to the provisioning and eviction strategy of spot instances to avoid hindering a new regular instance provisioning. 



\section{Conclusion}\label{sec:conclusion}
In this paper, we formally show that the spot instance can deliver additional benefits for both CSPs and customers under some conditions through a two-stage dynamic game modeling and equilibrium analysis. The two-stage dynamic game reflects the complex interplay between CSPs and customers, where we model various heterogeneities in customers' willingness to pay, on-demand and spot instance services' different QoS, and resource utilization. We believe that our paper provides a clear insight into a pricing policy for the spot instance service for increasing CSPs' revenue.

\begin{appendices}

\section{Proof of Lemma~\ref{lem:condition}}\label{proof:lem_condition}
\begin{proof}

Recall that we already have the customers' behavior analysis at Stage II as shown in Lemma~\ref{lem:stage2}.
\begin{eqnarray}
\label{eq:stage2}
a^{\star} (\theta) = 
\begin{cases}
o, \quad \text{if} \quad \frac{p_o - p_s}{q_o - q_s} < \theta \leq 1,\\
s, \quad \text{if} \quad \frac{p_s}{q_s} < \theta \leq \frac{p_o - p_s}{q_o - q_s}, \\
n, \quad \text{otherwise.}
\end{cases}
\end{eqnarray}

Now, we start the equilibrium analysis at Stage I. CSP decides the on-demand price and the spot price, $p_o$ and $p_s,$ to maximize its revenue, i.e.,

\begin{align*}
\bm p^\star (t) &= \arg \max_{\bm p(t)} \pi(\bm p(t)),\\ 
&= \arg \max_{\bm p(t)} \sum_{m \in \set M} \big( \pi_o^m(p_o^m(t)) + \pi_s^m(p_s^m(t)) \big),
\end{align*}
whereby the definition in~\eqref{eq:def_revenue}, the total revenue of CSP at time slot $t$ under cluster $m$ can be written by:
\begin{eqnarray*}
\pi^m(p_o^m(t), p_s^m(t))
&=& p_o^m (t) \cdot \int_0^1 x_{\theta}(t) \cdot  \bm 1_{\{\theta \in \Theta_o^m(t)\}} d \theta \cr
&+& p_s^m (t) \cdot \int_0^1 x_{\theta}(t) \cdot  \bm 1_{\{\theta \in \Theta_s^m(t)\}} d \theta.
\end{eqnarray*}
Because all the time slots and clusters are mutually independent, i.e., the pricing policy in specific $t$ and $m$ do not affect other time slots and clusters, thus we can find the equilibrium for each of them independently. In other words, we have
$$
(p_o^{m, \star}(t) , p_s^{m, \star}(t)) = \arg \max_{(p_o^m(t), p_s^m(t))} \pi^m (p_o^m(t), p_s^m(t)),
$$
where we will omit $t$ and $m$ in notation for simplicity.
Then, thanks to Definition~\ref{def:average resource utilization}, we get:
\begin{eqnarray*}
\pi &=& p_o \int_0^1 x_{\theta} \cdot  \bm 1_{\{\theta \in \Theta_o\}} d \theta 
+ p_s \int_0^1 x_{\theta} \cdot  \bm 1_{\{\theta \in \Theta_s\}} d \theta,\cr
&=&  p_o \gamma_o |\Theta_o| + p_s \gamma_s |\Theta_s|, \cr
&=& \gamma_o p_o  \big(1 - \frac{p_o - p_s}{q_o - q_s} \big) 
+ \gamma_s  p_s \big(\frac{p_o - p_s}{q_o - q_s}  - \frac{p_s}{q_s}\big). 
\end{eqnarray*}
We now show that under $\tt C1$, $\pi(p_o, p_s)$ is a strictly concave function for the price $ p \triangleq (p_o, p_s),$ which can be shown by the Hessian of $\pi(p)$ is negative definite, for which Hessian has (i) negative trace and (ii) positive determinant. At first, the is straightforward that the Hessian of $\pi$ has the negative trace:
\begin{eqnarray*}
\frac{\partial^2 \pi}{\partial p_o^2} = \frac{-2\gamma_o}{q_o - q_s} < 0, \cr
\frac{\partial^2 \pi}{\partial p_s^2} = \frac{-2 \gamma_s}{q_o - q_s} - \frac{2 \gamma_s}{q_s} < 0.
\end{eqnarray*}
Moreover, the determinant of Hessian is rewritten by:
\begin{flalign}\label{eq:determinant}
& \frac{\partial^2 \pi}{\partial p_o^2} \!\cdot\! \frac{\partial^2 \pi}{\partial p_s^2} - \frac{\partial^2 \pi}{\partial p_o \partial p_s} \!\cdot \!\frac{\partial^2 \pi}{\partial p_s \partial p_o} \cr 
& = \frac{-2\gamma_o}{q_o - q_s} \!\cdot\! \Big( \frac{-2 \gamma_s}{q_o - q_s} - \frac{2 \gamma_s}{q_s} \Big) - \Big( \frac{\gamma_o + \gamma_s}{q_o - q_s} \Big )^2 \cr
& = \frac{4\gamma_o \gamma_s}{(q_o - q_s)^2} + \frac{4\gamma_o \gamma_s}{q_s (q_o - q_s)} - \frac{\eta^2}{(q_o - q_s)^2} \cr
& = 4 \gamma_o \gamma_s q_s + 4 \gamma_o \gamma_s (q_o - q_s) - q_s \eta^2 \cr
& = 4\gamma_o \gamma_s q_o  - \eta^2 q_s.
\end{flalign}
Then, the determinant in~\eqref{eq:determinant} is positive under $\tt C1$, because
$$
\gamma_o > \eta q_s / 2 q_o \text{ and } \gamma_s > \eta/2 \Rightarrow \gamma_o \gamma_s > \eta^2 q_s / 4 q_o.
$$
Thus, there exists a unique equilibrium price $p = (p_o^\star, p_s^\star)$ under the condition $\tt C1.$ 
\end{proof}

\section{Proof of Theorem~\ref{thm:equilibrium}}\label{proof:thm}
\begin{proof}
We now find the equilibrium of a two-stage dynamic game between CSP and customers, which is defined in Section~\ref{subsec:game}. Because we showed that the revenue of CSP is strictly concave for price $p = (p_o, p_s)$ in Lemma~\ref{lem:stage2}, thus the equilibrium can be found at the point where the first derivatives of $\pi$ with respect to $p_o$ and $p_s$ are equal to zero, i.e.,
\begin{eqnarray}
\frac{\partial \pi}{\partial p_o} = 0 \;\; \text{and} \;\; \frac{\partial \pi}{\partial p_s} = 0.
\end{eqnarray}
The first derivative of $\pi$ with respect to $p_o$ is:
\begin{eqnarray*}
\frac{\partial \pi}{\partial p_o} =
-\frac{\gamma_o p_o}{q_o - q_s} + \gamma_o (1 - \frac{p_o - p_s}{q_o - q_s}) + \frac{\gamma_s p_s}{q_o - q_s}.
\end{eqnarray*}
Then, the derivative becomes zero if the following holds:
\begin{eqnarray}\label{eq:p_s}
p_s = \frac{\gamma_o}{\gamma_o + \gamma_s} (2 p_o - q_o + q_s).
\end{eqnarray}
Similarly, the first derivative of $\pi$ with respect to $p_s$ is:
\begin{eqnarray*}
\frac{\partial \pi}{\partial p_s} =
\frac{\gamma_o p_o}{q_o - q_s} - (\frac{\gamma_s p_s}{q_o - q_s} + \frac{\gamma_s p_s}{q_s}) + \gamma_s (\frac{p_o - p_s}{q_o - q_s} - \frac{p_s}{q_s}).
\end{eqnarray*}
Then, the derivative becomes zero if the following holds:
\begin{eqnarray}\label{eq:p_o}
p_o = \frac{2\gamma_s q_o p_s}{(\gamma_o + \gamma_s) q_s}.
\end{eqnarray}
Then, we can get the equilibrium price $ p^\star = (p_o ^\star , p_s^\star)$ as shown in~\eqref{eq:eq_price}, which satisfies both equations~\eqref{eq:p_s} and~\eqref{eq:p_o} concurrently.

Finally, we can derive the equilibrium market share of each service~\eqref{eq:eq_marketshare} from the equation in~\eqref{eq:prop_marketshare}. Consequently, the definition~\eqref{eq:def_revenue} leads to the equilibrium revenue~\eqref{eq:eq_revenue}. 
\end{proof}

\section{Proof of Lemma~\ref{lem:on-demand-only}}\label{proof:lem:on-demand-only}

\begin{proof}
We derive the equilibrium price, market share, and revenue when CSP only serves the on-demand instance service. By the backward induction, we first find the customers' behavior at Stage II:
\begin{eqnarray}
a^\star (\theta) =
\begin{cases}
o, \quad \text{if} \quad \frac{p_o}{q_o} < \theta \leq 1,\\
n, \quad \text{otherwise.}
\end{cases}
\end{eqnarray}
This is because of the fact that the customers decide to choose the on-demand instance service if the following holds:
\begin{eqnarray*}
\theta q_o  - p_o > 0.
\end{eqnarray*}
Next, we provide the equilibrium analysis at Stage I. CSP chooses its price strategy $p_o$ so as to maximize its revenue, i.e.,
$$
p_o^\star = \arg \max_{p_o} \pi(p_o),
$$
where the revenue of on-demand service is derived as:
\begin{eqnarray*}
\pi(p_o) = \gamma_o p_o \int_{\Theta_o} d \theta
 = \gamma_o p_o (1-\frac{p_o}{q_o}).
\end{eqnarray*}
Then, there exists a unique $p_o$ maximizing $\pi$, because $\pi$ is strictly concave with respect to $p_o.$ 
Thus, we get the equilibrium price $\bar p_o^\star = q_o /2$, which satisfies $d\pi/dp_o = 0.$ We can easily derive the equilibrium market share and the equilibrium revenue.
\end{proof}

\section{Proof of Proposition~\ref{prop:monopoly}}\label{proof:prop:monopoly}

\begin{proof}
    Now, we show the properties in Proposition~\ref{prop:monopoly}.
\smallskip

\begin{compactenum}[(\em a)]
    \item The price property is rewritten as:
    \begin{flalign}\label{eq:property_price}
        p_o^\star > \bar p_o^\star &\Leftrightarrow \frac{2\gamma_o \gamma_s q_o (q_o - q_s)}{4\gamma_o \gamma_s q_o - \eta^2 q_s} > \frac{q_o}{2}, \cr
        &\Leftrightarrow 4 \gamma_o \gamma_s (q_o - q_s) > 4 \gamma_o \gamma_s q_o - \eta^2 q_s, \cr
        &\Leftrightarrow \gamma_o \gamma_s < \eta^2 / 4.
    \end{flalign}
    Then, the inequality~\eqref{eq:property_price} always holds, because we defined $\eta = \gamma_o + \gamma_s$, thus the maximum value of $\gamma_o \gamma_s = \eta^2/4,$ where $\gamma_o = \gamma_s = \eta /2.$
    \item The first market share property is rewritten as:
    \begin{flalign} \label{eq:property_marketshare}
    |\Theta_o^\star| < |\bar \Theta_o^\star| &\Leftrightarrow  \frac{\gamma_s(2\gamma_o q_o - \eta q_s)}{4 \gamma_o \gamma_s q_o - \eta^2 q_s} < 1/2, \cr
    &\Leftrightarrow 2 \gamma_s (2\gamma_o q_o - \eta q_s) < 4 \gamma_o \gamma_s q_o - \eta^2 q_s, \cr
    &\Leftrightarrow  \gamma_s > \eta /2.
    \end{flalign}
    The inequality~\eqref{eq:property_marketshare} is straightforward by the condition $\tt C1$. The second market share property is rewritten as:
    \begin{flalign}\label{eq:property_marketshare2}
    & |\Theta_o^\star| + |\Theta_s^\star| > |\bar \Theta_o^\star| \cr
    & \Leftrightarrow \frac{\gamma_s(2\gamma_o q_o - \eta q_s)}{4 \gamma_o \gamma_s q_o - \eta^2 q_s} + \frac{\gamma_o(\gamma_s - \gamma_o)q_o}{4\gamma_o \gamma_s q_o - \eta^2 q_s} > 1/2, \cr
    &\Leftrightarrow 6\gamma_o \gamma_s q_o - 2 \gamma_s \eta q_s - 2 \gamma_o^2 q_o > 4 \gamma_o \gamma_s q_o - \eta^2 q_s, \cr
    &\Leftrightarrow (\gamma_s - \gamma_o) (2 \gamma_o q_o - \eta q_s) > 0.
    \end{flalign}
    Then, we get that \eqref{eq:property_marketshare2} holds, because $\gamma_s - \gamma_o > 0$ and $2 \gamma_o q_o - \eta q_ s> 0$ from the condition $\tt C1$.
    \item The revenue property is rewritten as:
    \begin{flalign}\label{eq:property_revenue}
    \pi^\star > \bar \pi^\star &\Leftrightarrow  \frac{\gamma_o^2 \gamma_s q_o(q_o - q_s)}{4\gamma_o \gamma_s q_o - \eta^2 q_s} > \frac{\gamma_o q_o}{4}, \cr
    &\Leftrightarrow 4 \gamma_o \gamma_s (q_o - q_s) > 4 \gamma_o \gamma_s q_o - \eta^2 q_s, \cr
    &\Leftrightarrow \gamma_o \gamma_s < \eta^2 /4.
    \end{flalign}
        Then, \eqref{eq:property_revenue} always holds, because it is the same as~\eqref{eq:property_price}.
    
\end{compactenum}
\end{proof}

\section{Proof of Proposition~\ref{prop:customer}}\label{proof:prop:customer}
By definition of utility~\eqref{eq:utility}, the aggregate utility of the on-demand instance service is derived as:
\begin{eqnarray*}
    \int_{\Theta_o} u_i^\star(\theta)~ d \theta 
    &=& \int_{\Theta_o} \theta q_o x_{\theta} - p_o^\star x_{\theta} ~~ d\theta, \cr
    &=& \Big[ \gamma_o ({q_o\theta^2}/{2} - p_o^\star \theta) \Big]_{\Theta_o^\star}, \cr
    &=& q_o \gamma_o \gamma_s (2 \gamma_o q_o - \eta q_s) \cr
     &\times&  \Big(\frac{\gamma_o q_s(\gamma_s - \gamma_o) + 2 \gamma_o \gamma_s (q_o + q_s) - \eta^2 q_s}{2(4\gamma_o \gamma_s q_o - \eta^2 q_s)^2}\Big),
\end{eqnarray*}
where the equilibrium price and the equilibrium of market share's interval are described in Theorem~\ref{thm:equilibrium} and Table~\ref{table}, respectively.
Similarly, the aggregate utility of the spot instance service can be derived from the equilibrium analysis results.
\end{appendices}

\bibliographystyle{IEEEtran}
\bibliography{main} 

\end{document}